\documentclass[letterpaper, 10 pt, conference]{ieeeconf}  % Comment this line out

\IEEEoverridecommandlockouts                        
\usepackage{amsmath}
\usepackage{amsfonts}
\usepackage{graphicx}
\usepackage{algorithm}
\usepackage{caption}
\usepackage{subcaption}
\usepackage[noend]{algpseudocode}
\usepackage{cite}
\usepackage{xcolor}

\newcommand{\norm}[1]{\left\lVert#1\right\rVert}

\newtheorem{proposition}{Proposition}

\newtheorem{lemma}{Lemma}
\newtheorem{remark}{Remark}

\newtheorem{assumption}{Assumption}

\title{\LARGE \bf
Optimal Minimax Mobile Sensor Scheduling Over a Network
}

\author{Samuel C. Pinto$^1$, Sean B. Andersson$^{1,2}$, Julien M. Hendrickx$^3$, and Christos G. Cassandras$^{2,4}$% <-this % stops a space
\\
$^1$Dept. of Mechanical Engineering, $^2$Division of Systems Engineering,\\ $^4$Dept. of Electrical and Computer Engineering \\
Boston University, Boston, MA 02215, USA \\
$^3$ICTEAM Institute, UCLouvain, Louvain-la-Neuve 1348, Belgium \\
\{samcerq,sanderss,cgc\}@bu.edu, julien.hendrickx@uclouvain.be
\thanks{This work was supported in part by NSF under grants ECCS-1931600, DMS-1664644, CNS-1645681, and CMMI-1562031, by ARPA-E's NEXTCAR program under grant DE-AR0000796, by AFOSR under grant FA9550-19-1-0158,  and by the MathWorks. The work of J. Hendrickx was supported by the “RevealFlight” Concerted Research Action (ARC) of the Federation Wallonie-Bruxelles, by the Incentive Grant for Scientific Research (MIS) “Learning from Pairwise Data” of the F.R.S.-FNRS.
}
}

\begin{document}

\maketitle
\thispagestyle{empty}
\pagestyle{empty}

%%%%%%%%%%%%%%%%%%%%%%%%%%%%%%%%%%%%%%%%%%%%%%%%%%%%%%%%%%%%%%%%%%%%%%%%%%%%%%%%
\begin{abstract}
We investigate the problem of monitoring multiple targets using a single mobile sensor, with the goal of minimizing the maximum estimation error among all the targets over long time horizons. The sensor can move in a network-constrained structure, where it has to plan which targets to visit and for how long to dwell at each node. We prove that in an optimal observation time allocation, the peak uncertainty is the same among all the targets. By further restricting the agent policy to only visit each target once every cycle, we develop a scheme to optimize the agent's behavior that is significantly simpler computationally when compared to previous approaches for similar problems.
\end{abstract} %%%%%The tar%%%%%%%get is constrained to move through a %%%%%%%%%%%%%%network.%%%%%%%%%%%%%%%%%%%%%%%%%%%%%%%%%%%%%%%%%%%%%%%%%%%%%%
\section{INTRODUCTION}
\label{sec:intro}
In many different applications, such as ocean temperature monitoring \cite{lan2013planning} or surveillance in smart cities \cite{Alam:2018ie}, a mobile agent equipped with sensors is responsible for monitoring the state of a system at multiple points of interest. In this setting, a usual goal is to minimize some metric of uncertainty of the monitored variables over long time horizons. We call this problem ``persistent monitoring" (PM). In terms of design, a common approach (see e.g. \cite{cassandras2013optimal} and \cite{Welikala2019P3}) is to consider the agent policy as being composed of two interdependent tasks: searching for an optimal agent trajectory and planning for how long to observe each of the targets. In this work, we model the PM problem as one of an agent moving on a network \cite{Welikala2019P3} and seek a periodic trajectory that ensures each target is visited, while optimizing a given uncertainty metric over long time horizons. PM is closely related to the well known Traveling Salesman Problem (TSP), as both problems try to minimize the time spent traveling between different targets. % (note that in PM, the agent cannot perform any sensing while traveling).
However, in PM, unlike TSP, one has to plan for how long to dwell at each target in order to minimize a metric of the targets state uncertainties (instead of simply the distance traveled).

In our previous work \cite{pinto2020periodicfull,pinto2019monitoring,pinto2020multidimensional}, we investigated the multi-agent version of this problem, with the goal of finding locally optimal movement agent policies that minimized the mean squared estimation error. {One major limitation of this previous work was that each step of the optimization (which was a gradient descent procedure) relied on the numerical integration of matrix ordinary differential equations in order to compute the gradients.} Integrating these equations is a major computational burden and makes the approach in \cite{pinto2020periodicfull} unsuitable for settings with a large number of targets.
%were a lack of notion of global optimality and the fact that it used a gradient-based technique to optimize the trajectory. We highlight that in the process of computing the gradients, numerical integration of matrix ordinary differential equations was necessary. Integrating these equations is a major computational burden and makes the approach in \cite{pinto2020periodicfull} unsuitable for settings with a large number of targets.

In the present paper, we constrain ourselves to considering a single agent in order to establish stronger results than in previous publications{, with the goal of generalizing these stronger results to multi-agent settings in future work.}
Moreover, instead of minimizing the mean squared estimation error, we consider a different optimization metric,
the minimization of the maximum uncertainty over different targets. 
In addition to simplifying the computational burden, this metric may be preferred in some settings. For example, when monitoring safety critical systems that cannot operate over a given threshold (for instance a maximum temperature), one wants to make sure that every target location does not exceed that {``worst case"} limit (as opposed to optimizing an ``average" of the chance of violating it). {Some examples of applications where a critical threshold on the state uncertainty should not be exceeded include monitoring of wildfire or faults in civil infrastructure systems using unmanned aerial vehicles \cite{lin2018kalman, shakhatreh2019unmanned}}. Considering this cost function, {the main contribution of this paper is to show that}, in an optimal schedule all the targets have the same peak uncertainty. A parallel can be drawn between this problem and resource allocation problems, where we often find results showing that all the users have the same utility in the optimal schedule (see e.g. \cite{luss1999equitable}). In this paper, the resource to be allocated is the time for which each target is observed. {We exploit this necessary condition and introduce a novel} gradient-free and computationally inexpensive scheme to optimize the time spent at each target. For visiting sequences where each target is visited only once, this approach provably gives the optimal time that should be spent at each target for a fixed cycle period $T$. This is especially meaningful in scenarios where targets are heterogeneous and, in order to ensure a balanced operation of the entire network of targets, they must be observed for significantly different times.

This paper is organized as follows. In Section \ref{sec:formulation} we formally formulate the problem that we investigate. Section \ref{sec:optimal_time_distribution} establishes a necessary condition for optimality, which is used to derive an optimization scheme of cycles where each target is visited only once, given in Section \ref{sec:optimization_single_visit}. Finally, Section \ref{sec:results} shares some simulation results and Section \ref{sec:conclusion} presents our conclusions and ideas for future work.

 %%%%%%%%%%%% guarantees thptimal allocationat that is the %%%%%%%%%%%%%o%%%%%%%%%%%%%%%%%%%%%%%%%%%%%%%%%%%%%%%%%%%%%%%%%%%%%%%
\section{PROBLEM FORMULATION}
\label{sec:formulation}
We consider a set $\mathcal{V}$ composed of $M$ nodes (which are interchangeably called targets) in a connected and undirected graph $\mathcal{G}=(\mathcal{V},\mathcal{E})$. The set $\mathcal{E}$ is formed by edges $\{(i,j):\ i,j\ \in \mathcal{V}\}$. Each edge has an associated cost $d_{i,j}$, % Note that if there is no direct path between target $i$ and $j$, then $w_{i,j}=\infty$. 
representing the time the agent takes to travel between nodes $i$ and $j$. 

We also assume that each target has an internal state ${\phi_i} \in \mathbb{R}^{L_i}$ with dynamics
\begin{equation}
    \label{eq:dynamics_phi}
    {{\dot{\phi}}_i(t) = A_i{\phi}_i(t) + B_iu_i(t)+ {w}_i(t)},\ \ \ \ i=1,\dots, M,
\end{equation}
where $w_i$ are mutually independent, zero mean, white, Gaussian distributed processes with $E[{w}_i(t){w}_i(t)']=Q_i$ and $Q_i$ is a positive definite matrix. {Since we focus on the estimation problem, we do not discuss the design of the control $u_i(t)$, but we assume that the estimation algorithm has access to the values of $u_i(t)$.}

The mobile agent can navigate through the graph $\mathcal{G}$. When the agent visits a given node $i$, it observes its internal state with a linear observation model given by
\begin{equation}
    \label{eq:observation_model_ij}
    {z}_{i}(t)=H_i{\phi}_i(t)+{v}_{i}(t), \qquad {v}_{i}(t)\sim \mathcal{N}(0,R_i),
\end{equation}
where $v_i$ is also assumed to be white and statistically independent of $v_j$ if $i\neq j$ and $w_k$, $\forall k$.

We can fully describe the behavior of the agent at its $N$ visits by a vector of indices of the targets visited by the agent $\mathcal{Y}=[y_1,y_2,...,y_N]$, $1\leq y_k \leq M$, and the time spent by the agent at each of these visits $\mathcal{T} = [t_1,...,t_N]$, $t_k\geq0$. Note that even though in general $\mathcal{Y}$ and $\mathcal{T}$ can be infinitely long, in this paper we restrict ourselves to periodic policies (as will be discussed in the next subsection) and thus $N$ is finite.

Considering models \eqref{eq:dynamics_phi} and \eqref{eq:observation_model_ij}, {the maximum likelihood estimator} $\hat{\phi}_i(t)$ for the internal state of the targets is a Kalman-Bucy Filter with equations given by
\begin{subequations}
	\begin{align}
    \dot{\hat{{\phi}}}_i(t)&=A_i\hat{{\phi}}_i(t)+{B_iu_i(t)} \nonumber \\&\qquad+\eta_i(t)\Omega_i(t){H}_i'{R}_i^{-1}\left({{z}}_i(t)-{H}_i\hat{{\phi}}_i(t)\right), \label{eq:estimator_dynamics}
    \\
    \dot{\Omega}_i(t) &= A_i\Omega_i(t)+\Omega_i(t)A_i'+Q_i-\eta_i(t)\Omega_i(t){G}_i\Omega_i(t), \label{eq:dynamics_omega} 
	\end{align}
\end{subequations}
where $\Omega_i$ is the covariance matrix {of the estimator $\hat{\phi}_i(t)$}, $G_i=H_i'R_i^{-1}H_i$ and $\eta_i(t)=1$ if target $i$ is observed at time $t$ and $\eta_i(t)=0$ otherwise.% \textcolor{blue}{We highlight that the evolution of the covariance matrix given in \eqref{eq:dynamics_omega} is independent of the control $u_i(t)$.}

The goal is to design an agent movement policy that minimizes 
the maximum (over time and over all the targets) of a weighted norm of the long term covariance matrix, i.e.,
\begin{equation}
    \label{eq:cost}
    J(\mathcal{Y},\mathcal{T}) = \max_{i\in\{1,..,M\}} \limsup_{t\rightarrow \infty} g_i(\norm{\Omega_i(t)}), 
\end{equation}
where the target-specific (possibly) nonlinear weighting function $g_i(\cdot)$ is a strictly increasing function with $g_i(0)=0$ and $\lim_{x \rightarrow \infty}g_i(x)=\infty$, and $\norm{\cdot}$ is a norm on the space of positive semi-definite matrices. A usual choice is to have $g_i(x)=x$ and $\norm{X}=\text{tr}(X)$. An optimal strategy has an associated cost $J^*$ given by
\begin{equation}
       J^*(\mathcal{Y}^*,\mathcal{T}^*) = \min_{\mathcal{Y},\mathcal{T}} J(\mathcal{Y},\mathcal{T}).
\end{equation}

\subsection{Periodic Policies}
In this work we restrict ourselves to periodic policies. Note that in periodic trajectories, the number of parameters necessary to describe the behavior of the agent does not increase with the time horizon. They are therefore suitable for infinite horizon analysis. Moreover, if a target is visited in the period, it will be visited infinitely often, with inter-visit time upper bounded by the period. This notion fits very well into the paradigm of persistent monitoring, since one of the high level goals is exactly to ensure this persistence of visits.

In order to analyze the infinite horizon behavior of the covariance matrix, we initially make the following natural assumption that ensures the uncertainty of the internal state $\phi_i$ can be finite over long time horizons.

\begin{assumption}
    The pair $(A_i,H_i)$ is detectable, $\forall \ i\in\{1,...,M\}$.
\end{assumption}

We now revisit the following Proposition, initially introduced in \cite{pinto2020periodicfull}:
\begin{proposition}
    \label{prop:unique_attractive_sol_riccati_eq}
    If $\eta_i(t)$ is $T$-periodic and $\eta_i(t) > 0$ for $t$ in some interval $[a,b]\in[0,T]$ with $b>a$, then, under Assumption 1, there exists a unique non-negative stabilizing $T$-periodic solution $\bar{\Omega}_i(t)$ to \eqref{eq:dynamics_omega}.
\end{proposition}

The consequence of this proposition is that, in a periodic schedule, visiting a given target for any finite amount of time is enough to guarantee that the covariance will converge to a limit cycle that does not depend on the initial conditions. %Moreover, if $A_i$ is unstable (or marginally stable), uncertainty will grow unbounded if this target is never visited.
In this paper we will focus on cases where $A_i$ is unstable, since it requires targets to be visited infinitely often to ensure a bounded cost, and thus fits well into the persistent monitoring framework. {Note that in some real applications, such as the surveillance of wildfires or the expansion of cracks in civil infrastructure \cite{lin2018kalman,shakhatreh2019unmanned}, the dynamics of the state are often unstable.}
%This is a direct consequence of the the general solution to a linear first order matrix ODE (note that \eqref{eq:dynamics_omega} becomes linear in this case). 
%Therefore, in this paper we focus on the case where $A_i$ is not stable, which implicitly contains the fact that targets always need to be visited persistently. However, the results of this paper can be extended to the more general case. 
In addition, to avoid the degenerate case where $Q_i$ is semidefinite such that covariance could potentially stay null in some direction, we assume that $Q_i$ is positive definite. These assumptions are now formalized:

\begin{assumption}
    \label{assump:instablity}
    $A_i$ is not stable and $Q_i$ is positive definite.
\end{assumption}

\begin{remark}

The optimization approach we develop in this paper could be easily adapted to also handle targets with stable $A_i$. However, if $A_i$ is allowed to be stable, it could be the case that the optimal action would be never to observe such a target, since its uncertainty would be bounded even if the agent never visits it. In this paper, due to space constraints, we mainly focus on the optimization of dwelling times $\mathcal{T}$ and { by assuming that $A_i$ is unstable, we have that the set of visited targets is fixed. In future work we plan to consider both stable and unstable $A_i$, along with criteria for the inclusion or exclusion of targets in the visiting sequence. A consequence of Assumption \ref{assump:instablity} is that every target should be visited in one period of the cycle, however the optimal visiting order is still to be determined.} %For the sake of consistency, we assume that $A_i$ is not stable, since the agent necessarily has to visit this type of targets persistently. 
%\textcolor{blue}{In future work, we plan to optimize visiting sequences $\mathcal{Y}$, In that setting, it will be important to consider both stable and unstable $A_i$.
\end{remark}
% On the converse side, if a target is never visited in a periodic schedule (i.e., $\eta_i(t)=0,\ \forall t\geq0$), its covariance matrix will converge to a steady state (constant) matrix if and only if the matrix $A_i$ is stable. This is a direct consequence of the the general solution to a linear first order matrix ODE (note that \eqref{eq:dynamics_omega} becomes linear in this case).

\section{PROPERTIES OF AN OPTIMAL POLICY}
\label{sec:optimal_time_distribution}
\subsection{Target's Perspective of a Periodic Policy}
In this subsection we briefly define some necessary notation. We acknowledge that this notation is complex, however, most of the terms we define here will be used only in the proofs of Props. \ref{prop:detivative_ton} and \ref{prop:derivative_toff}, both of which have a very intuitive interpretation.

We recall that the goal is to optimize the visiting sequence $\mathcal{Y}$ and dwelling times at the visited targets $\mathcal{T}$. One
way to interpret the vectors $\mathcal{Y}$ and $\mathcal{T}$ is that they capture the agent perspective of the trajectory. Each individual target, however, does not need all this information in order to fully determine the behavior of its uncertainty. From the target's perspective, all the information necessary to characterize the evolution of its covariance matrix is to define, for a single cycle of the periodic schedule, how long the target is observed for and the time spent between consecutive observations. Without loss of generality, we assume that the length $N$ of the vectors $\mathcal{Y}$ and $\mathcal{T}$ is one full cycle of the periodic trajectory.

We now discuss the conversion of indices from the agent's perspective to the target's.
First, for each target $i$, we group all the instances where this target was visited and define the vector $\mathcal{P}_i=[p_i^1,...,p_i^{N_i}]$, {where $N_i$ is the number of times target $i$ is visited in a cycle and $p_i^j$ is the position in the schedule $\mathcal{Y}$ that $i$ is visited the $j$-th time, i.e. $
p_i^j\in\{p |\ y_p=i,\ y_p\in{\mathcal{Y}}\}$ and $p_i^j<p_i^k$ if $j<k$}.
{Consider, for example, a visiting sequence of length $N=3$ given by $\mathcal{Y} = [1,2,1]$. Then, $N_1=2$, $N_2=1$, $\mathcal{P}_1=[p_1^1,p_1^2]=[1,3]$ and  $\mathcal{P}_2=[p_2^1]=[2]$.}
{Additionally, we define the tuple $(a(q),b(q))$ as the pair such that $p_{a(q)}^{b(q)}=q$. Hence, $a(q)$ is the target being visited at the agent's $q$-th visit and $b(q)$ represents how many times this target has been visited so far (including the current visit).}

Finally, we highlight some important timings and covariance matrix values at the steady state cycle, illustrated in Fig. \ref{fig:multiple_observations_same_target}.

\begin{figure}[h!]
    \centering
    \includegraphics[width=0.9\columnwidth]{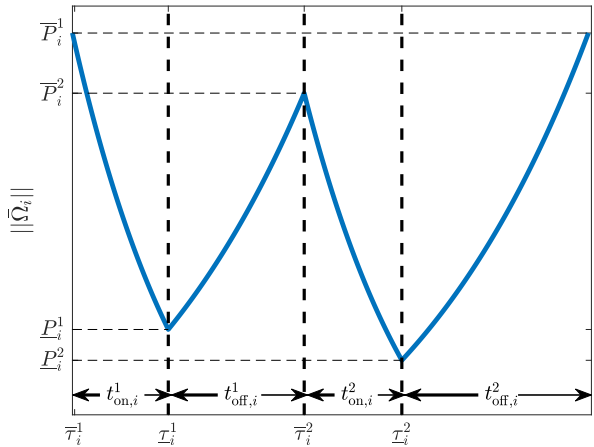}
    \caption{Temporal evolution of the steady state covariance matrix and waiting/observation times.}
    \label{fig:multiple_observations_same_target}
\end{figure}

In Fig. \ref{fig:multiple_observations_same_target}, $t_{\text{on},i}^k$ is the amount of time that target $i$ was observed when it was visited for the $k$-th time. $t_{\text{off},i}^k$ is the time spent between its $k$-th and $k+1$-th visits. These times are given by
\begin{subequations}
    \label{eq:targets_perspective_variables}
    \begin{equation}
        t_{\text{on},i}^k = t_{p_i^k},\qquad t_{p_i^k}\in\mathcal{T}, 
    \end{equation}
    \begin{equation}
        \label{eq:def_toff}
        t_{\text{off},i}^k= d_{i,a(p_i^k+1)}+\sum_{q=p_i^k+1}^{p_i^{k+1}-1}\left(t_{\text{on},a(q)}^{b(q)}+d_{a(q),a(q+1)}\right).
    \end{equation}
\end{subequations}
        
In \eqref{eq:def_toff}, $d_{i,a(c(k,i)+1)}$ is the travel time between target $i$ and the next target the agent visits. The index $q$ varies over all the visits the agent makes until it returns to target $i$ (note that $c(k,i)$ and $c(k+1,i)$ give the index of two consecutive visits to target $i$, from the agent's perspective). Moreover,  $t_{\text{on},a(q)}^{b(q)}$ is the time the agent spent at its $q$-th visit and $d_{a(q),a(q+1)}$ is the travel time between the agent's $q$-th and $(q+1)$-th visit.

We define $\overline{\tau}_i^k$ as the instant when the $k-th$ visit to target $i$ started and $\overline{P}_i^k$ its covariance at the beginning of that visit. Intuitively, variables with a bar over them refer to a local maximum peak ($\overline{P}_i^k$) and instant ($\overline{\tau}_i^k$). Similarly, $\underline{\tau}_i^k$ and $\underline{P}_i^k$ are respectively the time instant and the covariance at the end of the $k$-th visit and represent locally minimum peaks, as represented in Fig. \ref{fig:multiple_observations_same_target}. More formally:
\begin{subequations}
    \label{eq:targets_perspective_variables_2}
    \begin{equation}
        \overline{\tau}_i^k = \sum_{m=0}^{k-1}\left(t_{\text{on},i}^m+t_{\text{off},i}^m\right), \qquad \underline{\tau}_i^k =\overline{\tau}_i^k +t_{\text{on},i}^k,
    \end{equation}
        \begin{equation}
        \overline{P}_i^k = \bar{\Omega}(\overline{\tau}_i^k), \qquad\underline{P}_i^k = \bar{\Omega}(\underline{\tau}_i^k),
        \end{equation}
\end{subequations}
where $\bar{\Omega}_i(t)$ is the steady-state (periodic) covariance matrix.
Also, since both the agent trajectories and the steady state covariance are periodic, we have $t_{\text{on},i}^k=t_{\text{on},i}^{k+N_i}$, $t_{\text{off},i}^k=t_{\text{on},i}^{k+N_i}$, $\overline{P}_i^k=\overline{P}_i^{k+N_i}$ and $\underline{P}_i^k=\underline{P}_i^{k+N_i}$. For a given period $T$, the visiting instants are spaced by $T$, that is $\overline{\tau}_i^{k+N_i}=T+\overline{\tau}_i^{k}$ and $\underline{\tau}_i^{k+N_i}=T+\underline{\tau}_i^{k}$.

As noted previously, while this notation is cumbersome it is needed for the proofs of Props. \ref{prop:detivative_ton} and \ref{prop:derivative_toff}. Once those are established, only $t_{\text{on},i}^k$, $t_{\text{off},i}^k$ and $\overline{P}_i^k$ will be used in the remainder of the text.

\subsection{Necessary Condition for Optimality}

In this section, our main goal is to show that, for every visiting sequence $\mathcal{Y}$ that contains every target, an optimal allocation of dwelling times $\mathcal{T}$ must be such that $\limsup_{t\rightarrow\infty}g_i(\norm{\Omega_i(t)})$ is the same among all the targets. Towards this main goal of this section, first we claim the following auxiliary results.

\begin{lemma}
    \label{lemma:pd_nd_time_derivative}
    $\dot{\bar{\Omega}}_i(t) \prec 0$ if $\eta_i(t)=1$ and $\dot{\bar{\Omega}}_i(t) \succ 0$ if $\eta_i(t)=0$.
\end{lemma}

{The interpretation of this lemma is very intuitive: the steady state covariance decreases under sensing and increases otherwise.} The proof is a direct application of {Theorem 8.5} in \cite{lawson2006symplectic} and is omitted for space reasons.

\begin{lemma}
    \label{lemma:max_on_peak}
    If target $i$ is visited for a non degenerate amount of time, then $\limsup_{t\rightarrow \infty}g_i\left(\norm{{\Omega}_i(t)}\right)=\max_{1\leq k \leq N_i}g_i\left(\norm{\overline{P}_i^k}\right)$
\end{lemma}
\begin{proof}
    First, note that since $\Omega_i(t)$ converges to the bounded periodic function $\bar{\Omega}_i(t)$, $\limsup_{t\rightarrow \infty}g_i\left(\norm{{\Omega}_i(t)}\right)=\max_{0\leq t\leq T}g_i\left(\norm{\bar{\Omega}_i(t)}\right)$. For any time $t$ for which $\exists\  \epsilon>0$ such that $\eta_i(t+\epsilon)=0$, Lemma \ref{lemma:pd_nd_time_derivative} implies that $\bar{\Omega}_i(t+\epsilon)\succ \bar{\Omega}_i(t)$. 
    Conversely, if $\exists\  \epsilon>0$ $\eta_i(t-\epsilon)=1$, $\bar{\Omega}_i(t-\epsilon)\succ \bar{\Omega}_i(t)$.
    
    Therefore, the maximum value of $g_i\left(\norm{\bar{\Omega}_i}\right)$ can only occur in one of the instants when the target switches from not being observed ($\eta_i=0$) to being observed ($\eta_i=1$). The covariance at these instants is given by $\overline{P}_i^k$.
\end{proof}

Using the results of the two lemmas just established, we now show how the upper peak values $\overline{P}_i$ vary with $t_{\text{on},i}$ and $t_{\text{off},i}$. Note that here we consider the steady state covariance as a function exclusively of $t_{\text{on},i}$ and $t_{\text{off},i}$, since these parameters fully define the steady state behavior of the covariance matrix. Additionally, even though the proofs of the next two propositions are technical, their interpretations are intuitive: when a given target is observed for a longer time, its peak uncertainty will be lower. Conversely, if the time between observations increases, then the peak uncertainties will be higher.

\begin{proposition}
    \label{prop:detivative_ton}
    $\frac{\partial \overline{P}_i^k}{\partial t_{\text{on},i}^m} \prec 0$.
\end{proposition}
\begin{proof}
First, we define $\Phi_i^m$ as
\begin{equation}
    \Phi_i^m = \exp\left(\int_{0}^{t_{\text{on},i}^m}(A_i-
   \bar{\Omega}_i(t+\overline{\tau}_i^m)G_i)dt\right).
\end{equation}
To understand the intuition behind $\Phi_i^m$, recall Prop. 3 in \cite{pinto2020periodicfull} which states that the derivative $\frac{\partial {\bar{\Omega}}}{\partial t_{\text{on},i}^m}(t)$, for $t\in(\overline{\tau}_i^m,\underline{\tau}_i^m)$ is given by:
\begin{equation}
    \frac{\partial {\bar{\Omega}}}{\partial t_{\text{on},i}^m}(t) = \Sigma^T(t)\frac{\partial {\bar{\Omega}}}{\partial t_{\text{on},i}^m}(\overline{\tau}_i^m)\Sigma(t),
\end{equation}
where $\Sigma$ is the solution of the following ODE:
\begin{equation}
    \dot{\Sigma}(t) - (A-\bar{\Omega}_i(t)G_i)\Sigma(t) = 0,\ \Sigma(\underline{\tau}_i^m)=I.
\end{equation}
Therefore, $\Phi_i^m=\Sigma(\overline{\tau}_i^m)$ can be interpreted as the transition matrix between times $\overline{\tau}_i^m$ and  $\underline{\tau}_i^m$ of the homogeneous version of the ODE for which the derivative is a solution. For more details, we refer the reader to \cite{pinto2020periodicfull}.

Now, by computing the derivative of $\bar{\Omega}_i(t)$ as $t\rightarrow(\overline{\tau}_i^k)^-$ (recall that at time $t=\overline{\tau}_i^k$ the derivative is discontinuous, since at this instant the target switches from not being observed to being observed), we obtain the following recursive expression:
\begin{equation}
\frac{\partial {\bar{\Omega}}}{\partial t_{\text{on},i}^m}(\underline{\tau}_i^m) = \frac{\partial \underline{P}_i^m}{\partial t_{\text{on},i}^m} = (\Phi_i^m)^T\frac{\partial \overline{P}_i^{m}}{\partial t_{\text{on},i}^m}\Phi_i^m + \dot{\bar{\Omega}}_i((\underline{\tau}_i^{m})^-).
\end{equation}

Furthermore, defining $\Psi_i^m=\exp(A_i t_{\text{off},i}^m)$ and using again the result in Prop. 3 of \cite{pinto2020periodicfull}, we get:
\begin{equation}
    \label{eq:propagation_lower_upper_derivative}
    \frac{\partial \overline{P}_i^{m+1}}{\partial t_{on,i}^m} = (\Psi_i^m)^T\frac{\partial \underline{P}_i^{m}}{\partial t_{on,i}^m}\Psi_i^m.
\end{equation}
 
Now, repeating the same steps and propagating the previous expression to the $k$-th visit, $m\leq k \leq m+N_i-1$, we get the recursion:
\begin{multline}
    \frac{\partial \overline{P}_i^{k}}{\partial t_{\text{on},i}^m} =\\ (\Lambda_i^{k,m})^T\left(\frac{\partial \overline{P}_i^{m-1}}{\partial t_{\text{on},i}^m}+(\Phi_i^m)^{-T}\dot{\bar{\Omega}}_i((\underline{\tau}_i^{m})^-)(\Phi_i^m)^{-1}\right)\Lambda_i^{k,m},
\end{multline}
where $\Lambda_i^{k,m}=\prod_{\alpha=m}^{k}\Psi_i^\alpha\Phi_i^\alpha$. In particular, for $k=m+N_i-1$, due to periodicity, we have:

\begin{multline}
    \frac{\partial \overline{P}_i^{m}}{\partial t_{\text{on},i}^{m}} =  
    (\Lambda_i^{m+N_i-1,m})^T\left(\frac{\partial \overline{P}_i^{m}}{\partial t_{\text{on},i}^m}\right.\\\left.+(\Phi_i^m)^{-T}\dot{\bar{\Omega}}_i((\underline{\tau}_i^{m})^-)(\Phi_i^m)^{-1}\right)\Lambda_i^{m+N_i-1,m},
\end{multline}
which is a Lyapunov equation. Note that $\Lambda_i^{m+N_i-1,m}$ is stable, as discussed in Prop. 3 of \cite{pinto2020periodicfull}. Therefore all of its eigenvalues have modulus lower than one. Also, since $\Lambda_i^{m+N_i-1,m}$ is a product of matrix exponentials, its null space is trivial. This, along with the fact that Lemma \ref{lemma:pd_nd_time_derivative} tells us that $\dot{\bar{\Omega}}_i((\underline{\tau}_i^{m})^-) \prec 0$, implies that the Lyapunov equation has a unique negative definite solution and therefore $\frac{\partial \overline{P}_i^{m-1}}{\partial t_{\text{on},i}^{m}}\prec 0$.

Moreover, note that, for $m \leq k < m+N_i-1$,
\begin{equation}
    \frac{\partial \overline{P}_i^{m-1}}{\partial t_{on,i}^m} = (\Lambda_i^{m+N_i-1,k})^T\frac{\partial \overline{P}_i^{k}}{\partial t_{on,i}^m}\Lambda_i^{m+N_i-1,k},
\end{equation}
which leads us to conclude that, $\forall k$, $\frac{\partial \overline{P}_i^{k}}{\partial t_{on,i}^m}\prec 0.$
%Hence, $\frac{\partial \overline{P}_i^{k}}{\partial t_{on,i}^m}\prec 0$. Finally, we recall  that 
%from which we conclude that  $\frac{\partial \underline{P}_i^{k}}{\partial t_{on,i}^m}\prec 0$
\end{proof}

\begin{proposition}
    \label{prop:derivative_toff}
    $\frac{\partial \overline{P}_i^k}{\partial t_{\text{off},i}^m} \succ 0$.    
\end{proposition}
\begin{proof}
The proof is very similar to Prop \ref{prop:detivative_ton}. Therefore, only a brief will be given. Note that 
\begin{equation}
 \frac{\partial \overline{P}_i^m}{\partial t_{\text{off},i}^m} = (\Lambda_i^{m,m-1})^T\frac{\partial \overline{P}_i^{m-1}}{\partial t_{\text{off},i}^m}\Lambda_i^{m,m-1} + \dot{\bar{\Omega}}_i((\overline{\tau}_i^{m})^-).
\end{equation}
Using a similar recursion as in the previous proposition, we get
\begin{multline}
 \frac{\partial \overline{P}_i^{m-1}}{\partial t_{\text{off},i}^m} = (\Lambda_i^{m+N_i-2,m-1})^T\frac{\partial \overline{P}_i^{m-1}}{\partial t_{\text{off},i}^m}\Lambda_i^{m+N_i-2,m-1} \\+ (\Lambda_i^{m+N_i-2,m})^T\dot{\bar{\Omega}}_i((\overline{\tau}_i^{m})^-)\Lambda_i^{m+N_i-2,m}.
\end{multline}
Since $\dot{\bar{\Omega}}_i((\overline{\tau}_i^m)^-) \succ 0$, then $\frac{\partial \underline{P}_i^{m-1}}{\partial t_{\text{off},i}^m} \succ 0$ and consequently $\frac{\partial \overline{P}_i^{k}}{\partial t_{\text{off},i}^m} \succ 0$.
\end{proof}

Now that we have established some properties of the variation of the peaks of the steady state covariance matrix, we can explore the main result of this section. The next proposition can be interpreted analogously to resource allocation problems, where different targets are competing for the same resource $t_{\text{on},i}^k$. Therefore, an equilibrium (in the minimax sense) is reached when all the targets have the same utility.
However, unlike typical resource allocation problems, here the total resource $\sum_{i=1}^M \sum_{k=1}^{N_i}t_{\text{on},i}^k$ is not fixed. The reason why the total resource does not go to infinity is that increasing $t_{\text{on},i}^k$ to one target has an adverse effect to all other targets.

\begin{proposition}
    \label{prop:same_peak}
    Assuming a fixed sequence of visiting targets $\mathcal{Y}$, in an optimal allocation of visiting times $\mathcal{T}$ according to the cost \eqref{eq:cost}, the following must hold:
    $$\limsup_{t\rightarrow\infty}g_i(\norm{\Omega_i(t)})=\limsup_{t\rightarrow\infty}g_j(\norm{\Omega_j(t)}),$$
    with $i,\ j\in\{1,...,M\}.$
\end{proposition}
\begin{proof}
    %First, we recall the definition of $t_{{\text{off},i}}^m$ in \eqref{eq:def_toff}, where, since by hypothesis the the visiting sequence is fixed, $w_{i,j}$ should be thought as a constant. Then, 
    We prove by contradiction, showing that if the property given in the proposition does not hold, then there is a way to re-balance the observation times that is guaranteed to improve the performance. Suppose that for some target $i$
    \begin{equation}
        \label{eq:inequal_peak_hypothesis}
        g_i\left(\norm{\overline{P}_i^{\max} (t_{\text{on},i}^{1:N_i},t_{\text{off},i}^{1:N_i})}\right)<g_j\left(\norm{\overline{P}_j^{\max} (t_{\text{on},j}^{1:N_j},t_{\text{off},j}^{1:N_j})}\right),
    \end{equation}
    where the upper index $\max$ indicates that among all the peaks $P_i^k$, we pick the value of $k$ that yields the highest $g_i(\norm{P_i^k})$. We now propose to decrease the amount of time target $i$ is observed by $\epsilon$. Looking back to Eq. \eqref{eq:def_toff}, this implies that the waiting time between observations for all the other targets will also decrease. This updated policy generates a new set of observation times for target $i$, ($\tilde{t}_{\text{on},i}^{1:N_i}$), and updated waiting times between visits for all the other targets, ($\tilde{t}_{\text{off},j}^{1:N_j}$), while maintaining ${t}_{\text{off},i}^{1:N_i}$ and ${t}_{\text{on},j}^{1:N_j}$ constant.  There exists $\epsilon>0$ such that $\tilde{t}_{\text{on},i}^ k={t}_{\text{on},i}^k-\epsilon$ for some $k\in\{1,..,N_i\}$ and $\tilde{t}_{\text{off},j}^m<{t}_{\text{off},j}^m$ for some $m\in \{1,...,N_j\}$. Using Props. \ref{prop:detivative_ton} and \ref{prop:derivative_toff} we get that
    \begin{align}
        \overline{P}_i^{\max} (\tilde{t}_{\text{on},i}^{1:N_i},t_{\text{off},i}^{1:N_i}) &\succ  \overline{P}_i^{\max} ({t}_{\text{on},i}^{1:N_i},t_{\text{off},i}^{1:N_i}),\\
        \overline{P}_j^{\max} ({t}_{\text{on},j}^{1:N_j},\tilde{t}_{\text{off},j}^{1:N_j}) &\prec  \overline{P}_j^{\max} ({t}_{\text{on},j}^{1:N_j},t_{\text{off},j}^{1:N_j}).
    \end{align}
     Using the fact that both the norm and the derivative are continuous and strictly increasing, we can always pick an $\epsilon$ small enough such that the new peak of target $i$ is lower or equal to the new peak of $j$, i.e.,
    \begin{align}
        g_j\left(\norm{\overline{P}_j^{\max} ({t}_{\text{on},j}^{1:N_j},\tilde{t}_{\text{off},j}^{1:N_j})}\right) &<  g_j\left(\norm{\overline{P}_j^{\max} ({t}_{\text{on},j}^{1:N_j},t_{\text{off},j}^{1:N_j})}\right),\\
        g_i\left(\norm{\overline{P}_i^{\max} (\tilde{t}_{\text{on},i}^{1:N_i},t_{\text{off},i}^{1:N_i})}\right)&\leq g_j\left(\norm{\overline{P}_j^{\max} ({t}_{\text{on},j}^{1:N_j},\tilde{t}_{\text{off},j}^{1:N_j})}\right).
    \end{align}
    
    Since under the updated policy all the peaks $\overline{P}_j^m$, $1\leq m \leq N_j$, $1\leq j\leq M$, are lower for all the targets except target $i$, we recall Lemma \ref{lemma:max_on_peak} and conclude that this updated policy has a lower cost than the previous one. Therefore, the previous policy cannot be optimal, which proves the proposition.
    
    \end{proof}

    This proposition gives a necessary condition for optimally allocating observation times. Moreover, its constructive proof also gives insight on how to locally optimize the dwelling times for a fixed sequence. However, in general this property is not sufficient for optimally determining how much time the agent should spend at each target. In the remainder of this paper, we will restrict ourselves to a specific set of visiting sequences $\mathcal{Y}$ where the optimality property in Prop. \ref{prop:same_peak} can indeed be exploited to optimize the dwelling times at each target. %In this context, instead of investigating optimality among all the possible visiting sequences $\mathcal{Y}$, we will restrain ourselves to one specific type of visiting sequences: those where each target is observed only once in a cycle. For this set, we use the results we established in this section both to determine the optimal visiting sequence and to develop a practical and  computationally simple algorithm that can optimize the dwelling times.

    \section{Optimal Monitoring with Single Visit per Cycle}
    \label{sec:optimization_single_visit}
    In this section, we restrict ourselves to investigating visiting sequences $\mathcal{Y}$ for which each target is visited only once in a cycle. By making this assumption, we are able to prove that the optimal one with respect to the cost \eqref{eq:cost} is the Traveling Salesman Problem (TSP) cycle, which was intuitively expected. Additionally, we present a practical algorithm for optimizing for how long each target should be observed. As a side note, in this section we will omit the upper index of $t_{\text{on},i}$ and $t_{\text{off},i}$, since $N_i=1$.
    
    \begin{proposition}
        Among all the visiting sequences where each target is visited only once at each cycle, for any allocation of dwelling times $\mathcal{T}$, the TSP solution is the one that minimizes the cost \eqref{eq:cost}.
    \end{proposition}
    \begin{proof}
        We recall the $t_{\text{off},i}$ is completely defined by the amount of time the agent dwells in other targets plus the time it spends traveling between targets. If every target is visited only once, the TSP cycle is the one with the least amount of traveling time. Consequently, for fixed $t_{\text{on},i}$, the values of $t_{\text{off},i},\   \forall i$,  are minimized. This fact along with Proposition \ref{prop:derivative_toff} implies that for any set of observation times $\mathcal{T}$, the TSP tour will yield the lowest cost.
    \end{proof}
    
    \begin{remark}
        The problem of computing the TSP tour is NP-hard. However, efficient sub-optimal solutions are available (see e.g. \cite{Kirk2020}). The approach we will discuss for optimizing the dwelling times does not rely on having the optimal TSP tour, and indeed can handle any tour as long as every target  is visited exactly once. Therefore, if finding the TSP tour is computationally infeasible, we still can use a suboptimal TSP solution. Moreover, the results in this section can also be applied to connected graphs where it is not possible to visit every target if each node is constrained to be visited only once. {For instance, if a node is only connected to one other node in the network, the node that it is connected to will need to be visited at least twice. In this case, we can still use the results given in this section if we impose the constraint that all the visitation times to a given target, except one, are zero.}
    \end{remark}
    
    Having discussed the optimal visiting sequence in this particular subset of admissible $\mathcal{Y}$, we now approach the question of how to adjust the visitation time in a computationally simple way. The main idea behind our approach is to exploit the property that the upper peaks must coincide in an optimal solution. We develop a procedure to balance the observation times at each target, such that the peak uncertainties will coincide after convergence. {This iterative scheme, which updates the all the target
   's observation time $t_{\text{on},i}$ at each iteration, is given by:}
    \begin{equation}
        \label{eq:consensus_update_discrete}
        t_{on,i}[k+1] = t_{on,i}[k] +k_p\log\left(\frac{g_i\left(\norm{\overline{P}_i}\right)}{g_{\text{avg}}}\right),
    \end{equation}
    {where $g_{\text{avg}}=\left(\prod_{j=1}^Mg_j\left(\norm{\overline{P}_j}\right)\right)^{\frac{1}{M}}$ and $k_p$ is a small positive constant. It can be interpreted as a ``consensus" algorithm on the peak uncertainties $\overline{P}_i$, and thus its structure is very similar to geometric mean consensus algorithms \cite{FB-LNS}. The expression \eqref{eq:consensus_update_discrete} does not require the computation of gradients, which makes it computationally much less demanding than gradient-based approaches.}
    \begin{remark}
    At each iteration of \eqref{eq:consensus_update_discrete}, it is necessary to compute $\overline{P}_i$ and $g_{\text{avg}}$. For computing $\overline{P}_i$, we use the algorithm described in \cite{chu2004structure}, which converges quadratically and is numerically stable. Moreover, we note that $g_{\text{avg}}$ can be computed distributively if $\overline{P}_i$ is computed locally by each target and we assume that targets communicate among themselves using the same network connection structure that the agent uses to move (graph $\mathcal{G}$). Since this network is connected, consensus protocols can be employed at each iteration {in order to have each target compute} $g_{\text{avg}}$ separately.
    \end{remark}
    
    {In order to simplify the convergence analysis of this update law, we abstract it with the following differential equation:}
    \begin{equation}
        \label{eq:update_law_single_target}
        \frac{d}{dr}{t}_{\text{on},i} = k_p\log\left(\frac{g_i\left(\norm{\overline{P}_i}\right)}{g_{\text{avg}}}\right).
    \end{equation}
    {Note that this version of the update law considers continuous parameter variation, i.e., the auxiliary variable $r$ that should be understood as the continuous time equivalent of ``iteration index" and does not carry any ``time'' interpretation.}
    Moreover we assume that, at $r=0$, all the targets were observed for a strictly positive amount of time.
    We now prove its convergence:
    \begin{proposition}
        \label{prop:assymptotic_stability}
        Under the update law \eqref{eq:update_law_single_target}, the function $\max_{i,j}|g_i(\norm{P_i})-g_j(\norm{P_j})|$ is asymptotically stable.
    \end{proposition}
    \begin{proof}
        \begin{equation}
            \frac{d \overline{P}_i}{d r} = \frac{\partial \overline{P}_i}{\partial t_{\text{on},i}}\frac{dt_{\text{on},i}}{dr}+\frac{\partial \overline{P}_i}{\partial t_{\text{off},i}}\frac{dt_{\text{off},i}}{dr}.
        \end{equation}
        
        However, note that the sum of observation times among different targets is constant because 
        %\begin{equation}
        \begin{align}    
        \sum_{i=1}^N \frac{d{t}_{\text{on},i}}{dr} &=k_p\log \left(\frac{\prod_{j=1}^Mg_j\left(\norm{\overline{P}_j}\right)}{g_{\text{avg}}}\right) \nonumber\\&= k_pM\log \left(\frac{g_{\text{avg}}}{g_{\text{avg}}}\right)=0,
        \end{align}
         therefore under this control law the period $T$ is constant. Since $T=t_{\text{on},i}+t_{\text{off},i}$, we have $\frac{dt_{\text{on},i}}{dr}=-\frac{dt_{\text{off},i}}{dr}$ and
        \begin{equation}
            \frac{d \overline{P}_i}{d r} = \left(\frac{\partial \overline{P}_i}{\partial t_{\text{on},i}}-\frac{\partial \overline{P}_i}{\partial t_{\text{off},i}}\right)k_p\log\left(\frac{g_i\left(\norm{\overline{P}_i}\right)}{g_{\text{avg}}}\right).
        \end{equation}
        Consequently, if $g_i(\norm{\overline{P}_i})>g_{\text{avg}}$, $\frac{d \overline{P}_i}{d r} \prec 0.$ 
        Conversely, if $g_i(\norm{\overline{P}_i})<g_{\text{avg}}$, then $\frac{d \overline{P}_i}{d r} \succ 0.$ Note that
        \begin{equation}
            \max_{i,j}|g_i(\norm{\overline{P}_i})-g_j(\norm{\overline{P}_j})| = \max_{i}g_i(\norm{\overline{P}_i}) - \min_{j}g_j(\norm{\overline{P}_j}).
        \end{equation}
        Therefore since $g_{\text{avg}}$ is the geometric mean (i.e. its value is lower than the maximum and higher than the minimimum $g_j(\norm{\overline{P}_j})$), we get that
        \begin{equation}
            \frac{d}{dr}\max_{i}g_i(\norm{\overline{P}_i}) <0,\ \frac{d }{dr}\min_{j}g_j(\norm{\overline{P}_j})>0.
        \end{equation}
        Hence
        \begin{equation}
            \frac{d}{dr}\max_{i,j}|g_i(\norm{\overline{P}_i})-g_j(\norm{\overline{P}_j})|<0,
        \end{equation}
        which proves the proposition.
    \end{proof}
    \begin{remark}
        \label{remark:monoticity}
        In the proof of Prop. \ref{prop:assymptotic_stability}, we see that $\frac{d}{dr}\max_ig_i(\norm{\overline{P}_i})<0$. Therefore, this update law always reduces the cost defined in \eqref{eq:cost}. This also implies that if the targets do not have the same peak value $g_i\left(\norm{\overline{P}_i}\right)$, then the cost can be reduced by applying \eqref{eq:update_law_single_target}. 
    \end{remark}
    
    Now we show that the value achieved by update law \eqref{eq:update_law_single_target} is unique, i.e. does not depend on the the initial observation time distribution. 
    \begin{lemma}
        \label{lemma:unique_consensus}
        For a given period $T$ and a fixed visiting sequence $\mathcal{Y}$ where each target is visited once, there is a unique observation time distribution $\mathcal{T}$ such that  $g_i(\norm{\overline{P}_i})=g_j(\norm{\overline{P}_j})$, $\forall \ i,\ j.$
    \end{lemma}
    \begin{proof}
        Suppose there are two different costs $g_{con}$ and $g_{con}'$ such that all targets have the same peak value. Without loss of generality, we assume $g_{con}<g_{con}'$. Since the period $T$ is the same, $t_{\text{on},i}>t_{\text{on},i}'$ and $t_{\text{off},i}<t_{\text{off},i}'$  $\forall i$, due to Props. \ref{prop:detivative_ton} and \ref{prop:derivative_toff}. However, since the period is the same, we must have $\sum_i t_{\text{on},i}=\sum_it_{\text{on},i}'$, which yields a contradiction.
    \end{proof}
    
    Finally, we give a specialization of Prop. \ref{prop:same_peak} to the particular case we discuss in this section. 
    \begin{proposition}
        For a fixed visiting sequence $\mathcal{Y}$ and a given cycle period $T$, the observation times under the update law \eqref{eq:update_law_single_target} converge to the optimal observation time allocation $\mathcal{T}$ with respect to the cost function in \eqref{eq:cost}.
    \end{proposition}
    \begin{proof}
        For any set of observation times with period $T$, update law \eqref{eq:update_law_single_target} always reduces the cost while maintaining $T$ constant if $g_i(\norm{\overline{P}_i})$ is not the same for all the targets. Since there is a unique way such that every target has the same peak (and the update law \eqref{eq:update_law_single_target} ensures convergence to it), then the observation times after convergence of \eqref{eq:update_law_single_target} have to be optimal, otherwise the update law would be able to improve them.
    \end{proof}
    %Lemma \ref{lemma:unique_consensus}, along with Remark \ref{remark:monoticity}, lead us to conclude that the peak uncertainty is the same for all the targets in the general sense as stated in Prop. \ref{prop:same_peak}, but also that in the optimal distribution of visitation times considering a fixed period, the peaks coincide. 
    
    %Therefore, the scheme in  \eqref{eq:update_law_single_target} provides a way to find the optimal visiting sequence given a fixed period. To obtain optimal visiting times, one way is to search for an optimal period and at each step of the search the update process described in \eqref{eq:update_law_single_target} is used. 
    
    We have shown a simple way to optimize the observation times, given a fixed period $T$. However, we have not addressed the problem of choosing an adequate $T$. For choosing an optimized period, we use golden ratio search \cite{kiefer1953sequential}, that finds the global optimum in unimodal functions and local optima in a generic single variable function. 
    
    \begin{remark}
    While in all our simulations we have observed that the optimal cost as a function of $T$ is unimodal, proving that this function always has a unique minimum is the subject of ongoing work.
    \end{remark}
    %\subsection{Practical Considerations}
    %Note that the update law in  \eqref{eq:update_law_single_target} assumes continuous parameter variation and therefore is not immediately implementable in a digital setting. We then provide insights on the implementation of this update law. First, we note that we implement that law using the following iterative scheme:
    %\begin{equation}
    %    \label{eq:consensus_update_discrete}
    %    t_{on,i} \leftarrow t_{on,i} +k_p\log\left(\frac{g_i\left(\norm{\overline{P}_i}\right)}{g_{\text{avg}}}\right),
    %\end{equation}
%    where we assume that $k_p$ is very small. 

    The golden ratio search procedure is given in Algorithm \ref{alg:compute_optimal_visitation_times}. The function $f(T)$ corresponds to running the update law in \eqref{eq:consensus_update_discrete} until convergence. The value $f(T)$ is the value of $g_{\text{avg}}$ at the final iteration.
    \begin{algorithm}
\caption{Optimal Visitation Times Computation}
\label{alg:compute_optimal_visitation_times}
    \begin{algorithmic}[1]
    %\Procedure{GoldenRatioSearch}{}
    \State{\bf Input: $T_{\min}$, $T_{\max}$}.
    \State{$T_1 \leftarrow T_{\max} - (T_{\max}-T_{\min})/g$}
    \State{$T_2 \leftarrow T_{\min} + (T_{\max}-T_{\min})/g$}
    \While{$|f(T_2 )-f(T_1 )|<\epsilon$}
    \If{$f(T_2)>f(T_1)$}
    \State{$T_{\max}\leftarrow T_2$}
    \Else{$\ T_{\min}\leftarrow T_1$}
    \EndIf
    \State{$T_1 \leftarrow T_{\max} - (T_{\max}-T_{\min})/g$}
    \State{$T_2 \leftarrow T_{\min} + (T_{\max}-T_{\min})/g$}
    \EndWhile
    \State{\bf Return: $(T_1+T_2)/2$}
\end{algorithmic}
\end{algorithm}
    \begin{figure*}[htp!]
    \centering
    \begin{subfigure}[t]{0.31\textwidth}
        \centering\includegraphics[width=\textwidth]{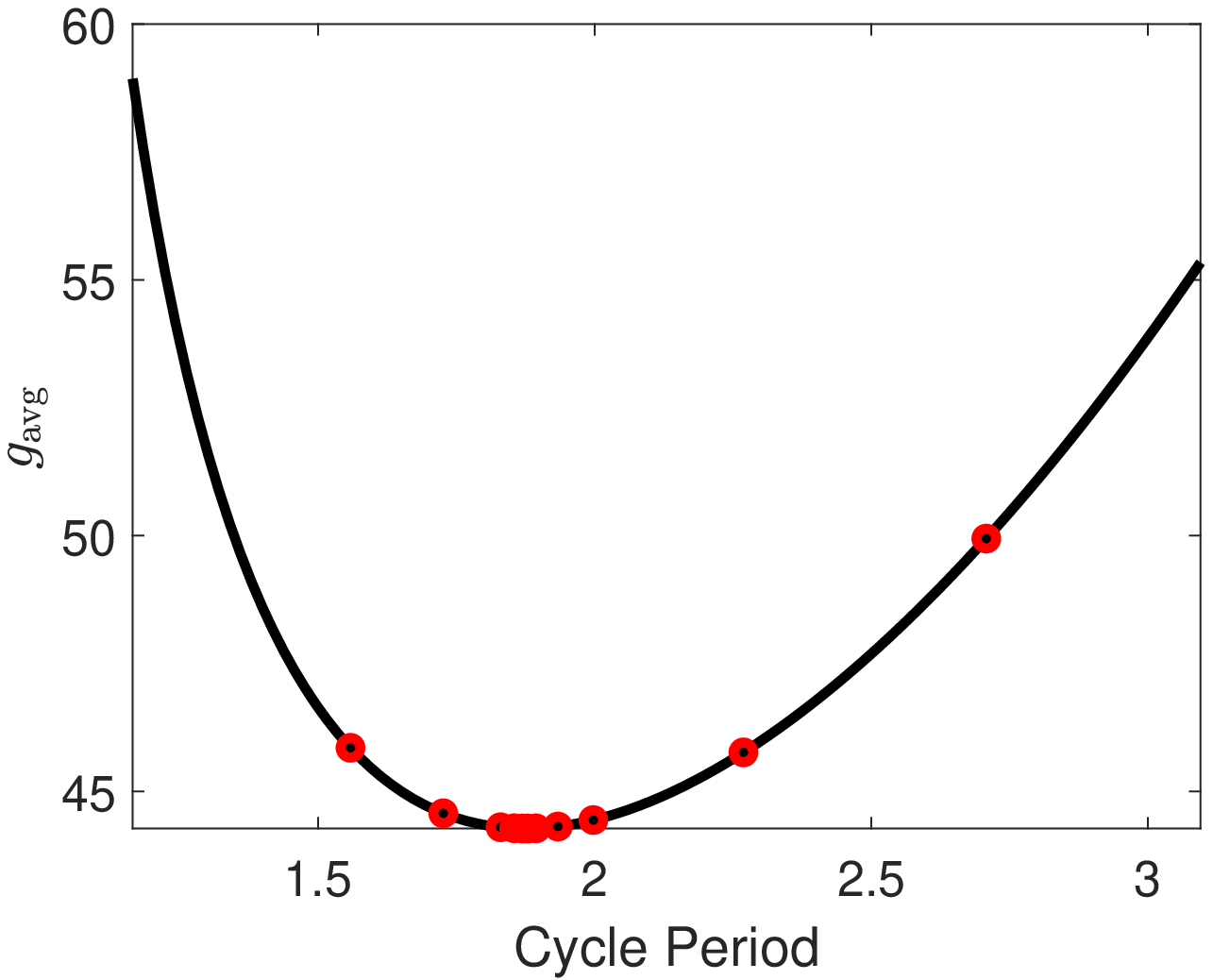}
        \caption{Peak value after balance among targets, as a function of cycle period.}
        \label{fig:peak_period}
    \end{subfigure}\hfill
    \begin{subfigure}[t]{0.32\textwidth}
        \centering\includegraphics[width=\textwidth]{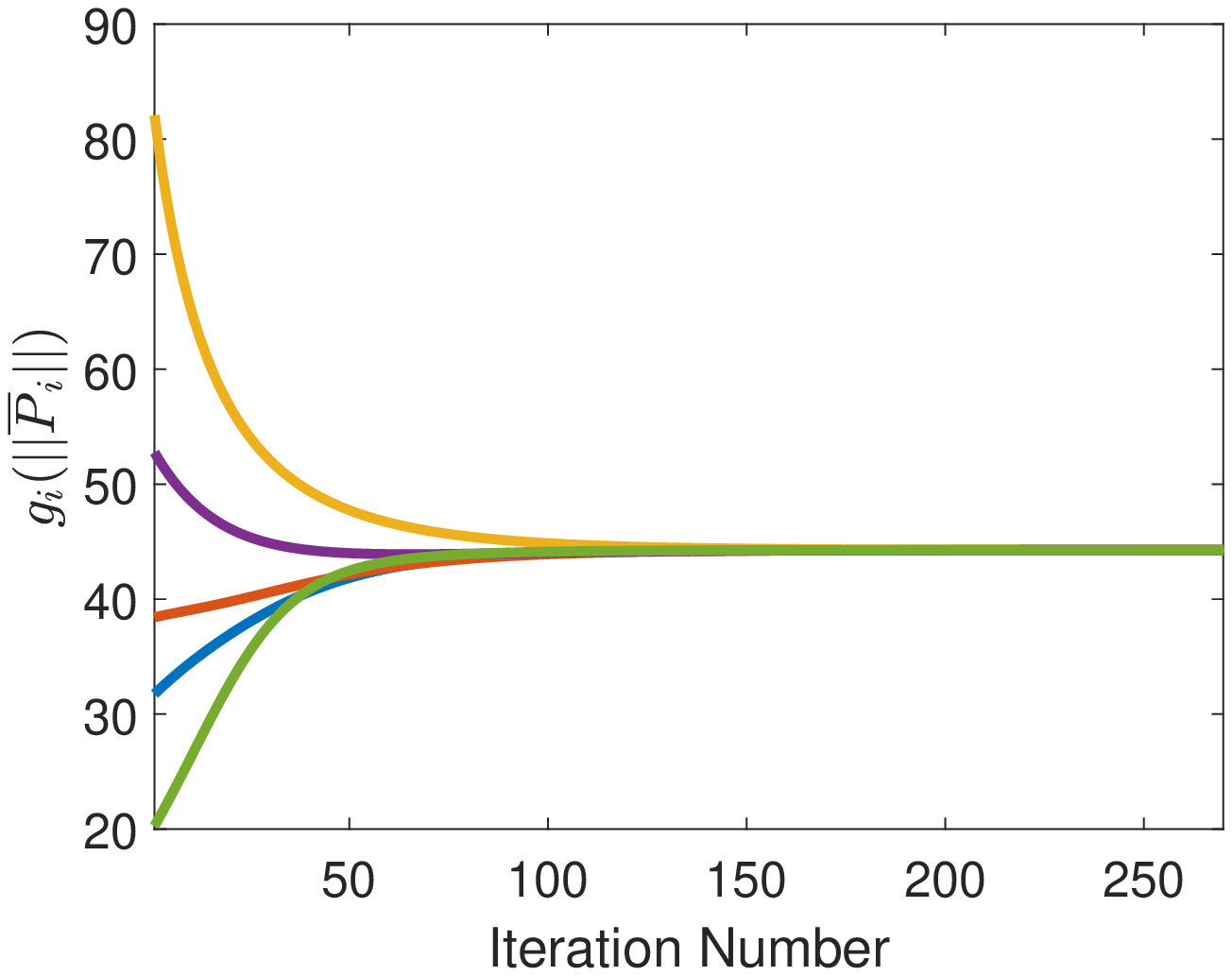}
        \caption{Peak uncertainty at the optimal period.}
        \label{fig:cost_iteration}
    \end{subfigure}\hfill
    \begin{subfigure}[t]{0.32\textwidth}
        \centering\includegraphics[width=\textwidth]{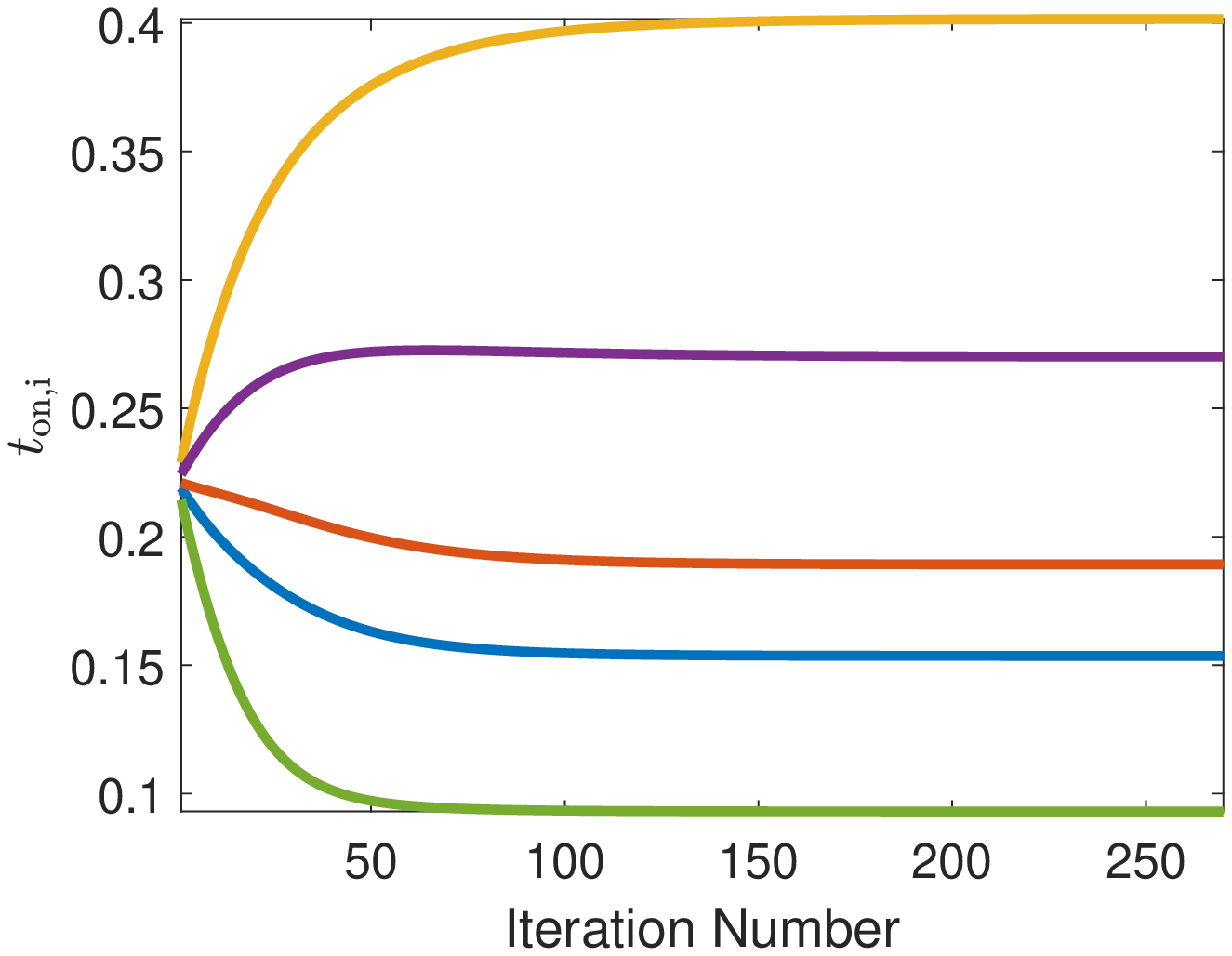}
        \caption{Observation time at the optimal period.}
        \label{fig:time_iteration}
    \end{subfigure}
    \caption{Results of simulating the described algorithm. In (a), the balanced peak uncertainty, as a function of the total cycle period. The red dots mark the values of $T$ that were explored by the golden ratio search. In (b)-(c), we show the the evolution of the peak uncertainty and the observation time for each target at the optimal period.}
    \label{fig:results_1_one_target}
    \end{figure*}

    \section{Simulation Results}
    \label{sec:results}
    We have simulated the model described in \eqref{eq:dynamics_phi} and \eqref{eq:observation_model_ij}, with parameters indicated in Table \ref{tab:sim_param}. Note that we assigned colors to targets in order to identify them in the figures. For simplicity, the internal states of the targets were assumed to be scalars. The target locations were drawn from a uniform distribution in $[0,0.5]\times[0,0.5]$. The targets location are displayed in Fig. \ref{fig:traj_simulation}. The graph was assumed to be fully connected with edge costs being the Euclidean distance between two targets. For the definition of the optimization goal as in \eqref{eq:cost}, we used $g_i(\xi)=\xi$, $\forall i$, and $\norm{\Gamma}=|\Gamma|$.
\begin{table}[h!]
\centering
\caption{Parameters for in the simulation.}
\label{tab:sim_param}
\begin{tabular}{ |c||c|c|c|c|c| } 

 \hline
 Target & 1 & 2 & 3 & 4 & 5  \\ 
 \hline
 Color& blue & red & yellow & purple & green  \\ 
 \hline
 $A_i$ & 0.3487 & 0.1915 & 0.4612 & 0.2951 & 0.1110  \\
 \hline
 $Q_i$ & 1.1924 & 1.2597 & 0.8808 & 1.7925 & 0.4363  \\
 \hline
 $R_i$ & 2.3140 & 7.1456 & 4.2031 & 5.2866 & 7.5314  \\
 \hline
\end{tabular}
%\end{center}
\end{table}

    For the visiting order, we used the optimal TSP tour. Note that the number of targets in this example is small enough that computing the TSP tour is computationally feasible. Then, Algorithm \ref{alg:compute_optimal_visitation_times} was deployed, using $k_p$ in \eqref{eq:consensus_update_discrete} set to $10^{-2}$ and $[T_{\min},T_{\max}]=[0.1t_{\text{travel}},3t_{\text{travel}}]$, where $t_{\text{travel}}$ is the total travel time in one period of the TSP cycle.% (i.e., time spent moving between different targets). 
    
        \begin{figure}[htp!]
    \centering
        \begin{subfigure}[t]{0.55\columnwidth}
    \vskip 0pt
    \includegraphics[width=\columnwidth]{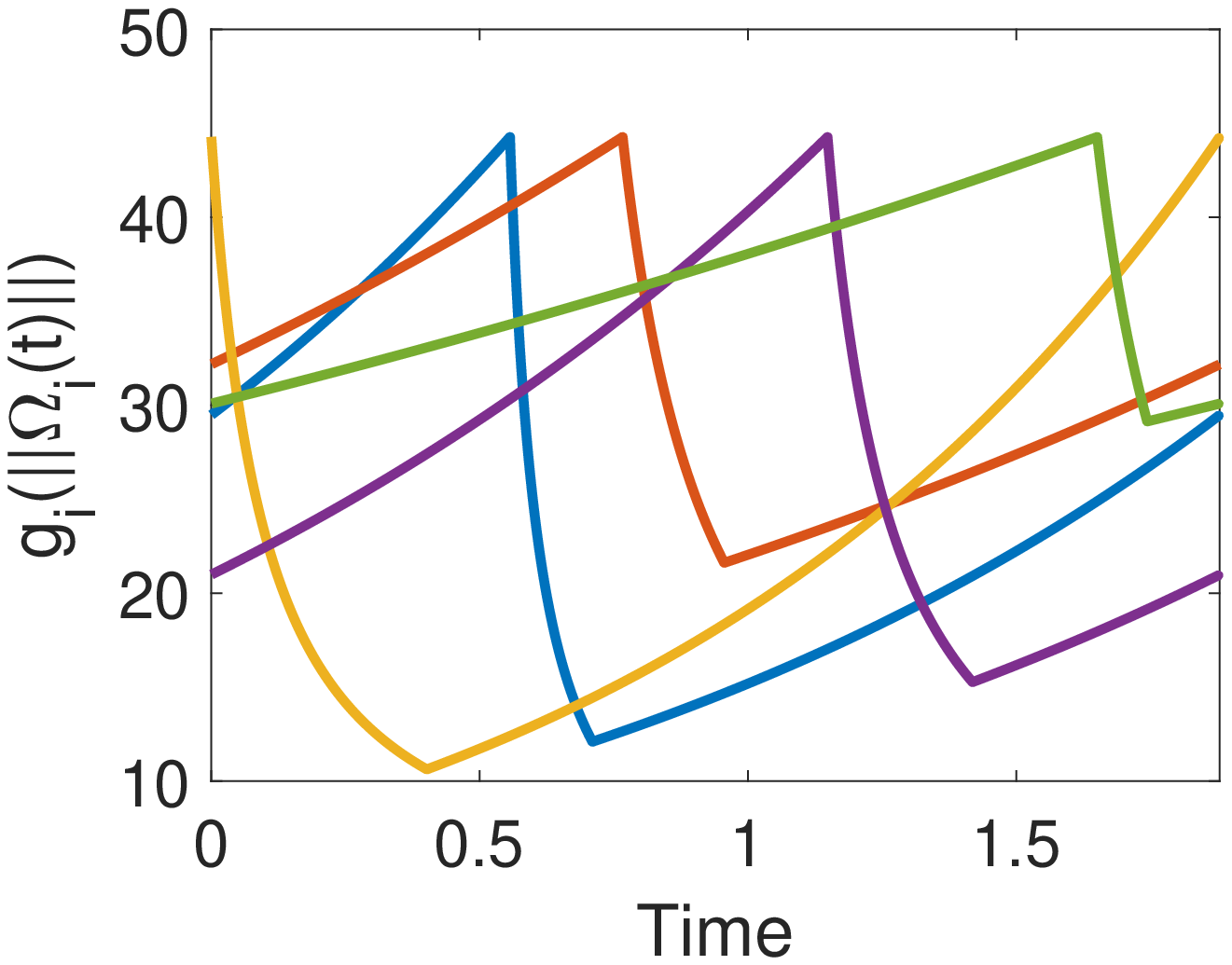}%
    \caption{Covariance over one period.}
    \end{subfigure}%
    \begin{subfigure}[t]{0.45\columnwidth}
    \vskip 5pt
    \includegraphics[width=\columnwidth]{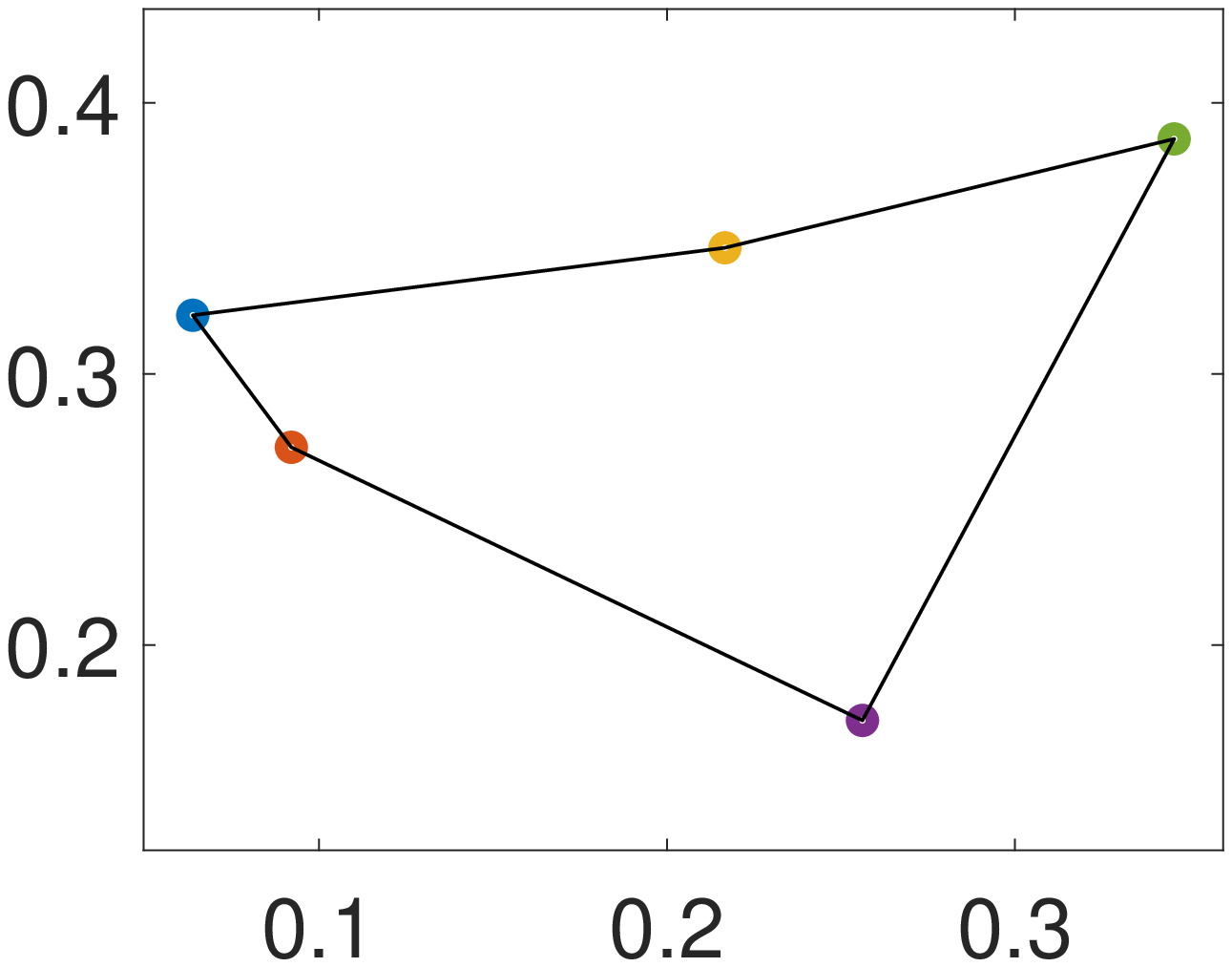}%
    \caption{Agent trajectory (black) and targets location (colored).}
    \label{fig:traj_simulation}%
    \end{subfigure}
    \caption{Results after optimizing the trajectory and dwelling times.}
    \label{fig:results_2_one_target}
    \end{figure}
    
    The results are shown in Figs. 
    \ref{fig:results_1_one_target} and \ref{fig:results_2_one_target}. In particular, Fig. \ref{fig:results_1_one_target} shows the evolution of the norm of the steady-state covariance matrix over one complete period of the agent trajectory, with the optimal distribution of observation times. In Fig. \ref{fig:results_1_one_target}, details of the optimization process are highlighted. In Fig. \ref{fig:peak_period} we can see how the uncertainty behaved as a function of the period (after balancing the observation times among targets by running \eqref{eq:consensus_update_discrete} until convergence). Moreover, this figure also highlights the fact that the golden ratio search scheme was able to efficiently converge to a local minimum. Figs. \ref{fig:cost_iteration} and \ref{fig:time_iteration} show how the observation time and the peak covariance varied while using the update law \eqref{eq:update_law_single_target}. Initially, all targets were visited for the same amount of time. However, as the iteration number increases, some targets are observed for more time than others. As expected, in the final iteration, all the peak covariances have converged to the same value. {In Fig. \ref{fig:results_2_one_target}, we can see the agent trajectory and that the targets indeed achieve the same steady state peak uncertainty.}

    \section{Conclusion and Future Work}
    \label{sec:conclusion}
    In this paper, we have derived a necessary condition for the optimal allocation of observation times of targets in a network being monitored by a mobile agent. In particular, for the set of policies where each target is observed only once, we were able to claim that the optimal visiting sequence is the TSP tour and we derived a computationally cheap and derivative free procedure to optimally allocate visiting times among different targets. This allows for scalable optimization of the monitoring performance considering heterogeneous targets.
    
    In future work, we want to investigate whether we can find a similar procedure to allocate observation times for trajectories where targets are allowed to be visited multiple times. Moreover, we intend to develop greedy schemes to efficiently search over different visiting sequences. Finally, we also intend to extend this approach to multi-agent scenarios.

\bibliographystyle{IEEEtran}
\bibliography{references.bib}

\begin{thebibliography}{10}
\providecommand{\url}[1]{#1}
\csname url@samestyle\endcsname
\providecommand{\newblock}{\relax}
\providecommand{\bibinfo}[2]{#2}
\providecommand{\BIBentrySTDinterwordspacing}{\spaceskip=0pt\relax}
\providecommand{\BIBentryALTinterwordstretchfactor}{4}
\providecommand{\BIBentryALTinterwordspacing}{\spaceskip=\fontdimen2\font plus
\BIBentryALTinterwordstretchfactor\fontdimen3\font minus
  \fontdimen4\font\relax}
\providecommand{\BIBforeignlanguage}[2]{{%
\expandafter\ifx\csname l@#1\endcsname\relax
\typeout{** WARNING: IEEEtran.bst: No hyphenation pattern has been}%
\typeout{** loaded for the language `#1'. Using the pattern for}%
\typeout{** the default language instead.}%
\else
\language=\csname l@#1\endcsname
\fi
#2}}
\providecommand{\BIBdecl}{\relax}
\BIBdecl

\bibitem{lan2013planning}
X.~Lan and M.~Schwager, ``{Planning Periodic Persistent Monitoring Trajectories
  for Sensing Robots in Gaussian Random Fields},'' in \emph{IEEE Int. Conf. on
  Robotics and Automation}.\hskip 1em plus 0.5em minus 0.4em\relax IEEE, 2013,
  pp. 2415--2420.

\bibitem{Alam:2018ie}
T.~Alam, G.~M. Reis, L.~Bobadilla, and R.~N. Smith, ``{A Data-Driven Deployment
  Approach for Persistent Monitoring in Aquatic Environments},'' in \emph{IEEE
  International Conference on Robotic Computing}, 2018, pp. 147--154.

\bibitem{cassandras2013optimal}
C.~G. Cassandras, X.~Lin, and X.~Ding, ``{An Optimal Control Approach to the
  Multi-agent Persistent Monitoring Problem},'' \emph{IEEE Transactions on
  Automatic Control}, vol.~58, no.~4, pp. 947--961, 2013.

\bibitem{Welikala2019P3}
S.~Welikala and C.~G. Cassandras, ``{Asymptotic Analysis for Greedy
  Initialization of Threshold-Based Distributed Optimization of Persistent
  Monitoring on Graphs},'' in \emph{Proc. of 21st IFAC World Congress}, 2020.

\bibitem{pinto2020periodicfull}
S.~C. Pinto, S.~B. Andersson, J.~M. Hendrickx, and C.~G. Cassandras,
  ``{Multi-Agent Persistent Monitoring of Targets with Uncertain States},'' in
  \emph{Arxiv}.\hskip 1em plus 0.5em minus 0.4em\relax Available online, 2020.

\bibitem{pinto2019monitoring}
------, ``{Optimal Multi-Agent Persistent Monitoring of the Uncertain State of
  a Finite Set of Targets},'' in \emph{Proc. of 58th Control and Decision
  Conference}, 2019.

\bibitem{pinto2020multidimensional}
------, ``{Multi-Agent Infinite Horizon Persistent Monitoring of Targets with
  Uncertain States in Multi-Dimensional Environments},'' in \emph{IFAC World
  Congress}, 2020.

\bibitem{lin2018kalman}
Z.~Lin, H.~H. Liu, and M.~Wotton, ``{Kalman Filter-based Large-scale Wildfire
  Monitoring with a System of UAVs},'' \emph{IEEE Transactions on Industrial
  Electronics}, vol.~66, no.~1, pp. 606--615, 2018.

\bibitem{shakhatreh2019unmanned}
H.~Shakhatreh, A.~H. Sawalmeh, A.~Al-Fuqaha, Z.~Dou, E.~Almaita, I.~Khalil,
  N.~S. Othman, A.~Khreishah, and M.~Guizani, ``{Unmanned Aerial Vehicles
  (UAVs): A Survey on Civil Applications and Key Research Challenges},''
  \emph{Ieee Access}, vol.~7, pp. 48\,572--48\,634, 2019.

\bibitem{luss1999equitable}
H.~Luss, ``{On Equitable Resource Allocation Problems: A Lexicographic Minimax
  Approach},'' \emph{Operations Research}, vol.~47, no.~3, pp. 361--378, 1999.

\bibitem{lawson2006symplectic}
J.~Lawson and Y.~Lim, ``{The Symplectic Semigroup and Riccati Differential
  Equations},'' \emph{Journal of Dynamical and Control Systems}, vol.~12,
  no.~1, pp. 49--77, 2006.

\bibitem{Kirk2020}
J.~Kirk, ``{Traveling Salesman Problem - Genetic Algorithm},'' 2020.

\bibitem{FB-LNS}
\BIBentryALTinterwordspacing
F.~Bullo, \emph{{Lectures on Network Systems}}, 1st~ed.\hskip 1em plus 0.5em
  minus 0.4em\relax Kindle Direct Publishing, 2020, with contributions by J.
  Cortes, F. Dorfler, and S. Martinez. [Online]. Available:
  \url{http://motion.me.ucsb.edu/book-lns}
\BIBentrySTDinterwordspacing

\bibitem{chu2004structure}
E.-W. Chu, H.-Y. Fan, W.-W. Lin, and C.-S. Wang, ``{Structure-preserving
  Algorithms for Periodic Discrete-time Algebraic Riccati Equations},''
  \emph{International Journal of Control}, vol.~77, no.~8, pp. 767--788, 2004.

\bibitem{kiefer1953sequential}
J.~Kiefer, ``{Sequential Minimax Search for a Maximum},'' \emph{Proceedings of
  the American Mathematical Society}, vol.~4, no.~3, pp. 502--506, 1953.

\end{thebibliography}
\end{document}